\documentclass[dvipdfmx] {article}

\usepackage{color}

\usepackage{ascmac}

\newcommand{\beq}{\begin{eqnarray*}}
\newcommand{\eeq}{\end{eqnarray*}}
\makeatletter
\renewcommand{\theequation}{\thesection.\arabic{equation}}
\@addtoreset{equation}{section}
\def\eqnarray{%
\stepcounter{equation}%
\let\@currentlabel=\theequation
\global\@eqnswtrue
\global\@eqcnt\z@
\tabskip\@centering
\let\\=\@eqncr
$$\halign to \displaywidth\bgroup\@eqnsel\hskip\@centering
$\displaystyle\tabskip\z@{##}$&\global\@eqcnt\@ne
\hfil$\displaystyle{{}##{}}$\hfil
&\global\@eqcnt\tw@$\displaystyle\tabskip\z@{##}$\hfil
\tabskip\@centering&\llap{##}\tabskip\z@\cr}
\makeatother
\newtheorem{theorem}{Theorem}[section]
\newtheorem{lemma}[theorem]{Lemma}

\newtheorem{proposition}[theorem]{Proposition}
\newtheorem{remark}{Remark}[section]

\newsavebox{\toy}
\savebox{\toy}{\framebox[0.65em]{\rule{0cm}{1ex}}}
\newcommand{\QED}{\usebox{\toy}}
\def\nlni{\par\ifvmode\removelastskip\fi\vskip\baselineskip\noindent}
\newenvironment{proof}{\nlni\begingroup\it Proof.\rm}{
\endgroup\vskip\baselineskip}

\begin{document}
\setlength{\baselineskip}{15pt}
\title{
Fluctuation 
of density of states for 1d Schr\"odinger operators
}
\author{
Fumihiko Nakano
\thanks{
Department of Mathematics,
Gakushuin University,
1-5-1, Mejiro, Toshima-ku, Tokyo, 171-8588, Japan.
e-mail : 
fumihiko@math.gakushuin.ac.jp}
}
\maketitle
\begin{abstract}
We consider 
the 1d Schr\"odinger operator with random decaying potential and compute the 2nd term asymptotics of the density of states, which shows substantial differences between the cases 
$\alpha > \frac 12$, 
$\alpha < \frac 12$
and 
$\alpha = \frac 12$. 
\end{abstract}


\section{Introduction}
1d Schr\"odinger operators 
with random decaying potentials have rich spectral structures depending the decay rate of the potential, so that there are many studies on this topic 
(e.g, \cite{KLS} and references therein).
In this paper 
we consider the following operator : 
\beq
H :=
- \frac {d^2}{d t^2} + a(t) F(X_t) 
\;
\mbox{ on }
L^2({\bf R})
\eeq
where 
$a \in C^{\infty}({\bf R})$, 
$a(-t) = a(t)$, $t>0$, 
$a$ 
is non-decreasing for 
$t > 0$, 
and
$a(t) = t^{- \alpha} (1 + o(1))$, $t \to \infty$, 
$\alpha > 0$.
$F \in C^{\infty}(M)$ 
on a torus $M$
such that 
\[
\langle F \rangle
:=
\int_M F(x) dx = 0, 
\]
and 
$( X_t )_{t \in {\bf R}}$
is a Brownian motion on 
$M$ 
with generator 
$L$. 
The spectrum of 
$H$ 
on 
$[0, \infty)$ 
is a.c. for 
$\alpha > \frac 12$, 
pure point for 
$\alpha < \frac 12$, 
and for 
$\alpha = \frac 12$, 
pure point on 
$[0, E_c]$ 
and 
s.c. on 
$[E_c, \infty)$ 
for some 
$E_c \ge 0$
\cite{KU}. 
For the 
level statistics problem, 
the point process 
$\xi_L$,  
whose atoms are composed of the rescaled eigenvalues of the finite volume restriction of 
$H$, 
converges to clock process 
($\alpha > \frac 12$), 
Sine$_{\beta}$-process 
($\alpha = \frac 12$), 
and 
Poisson process
($\alpha < \frac 12$)
\cite{KN1, N2, KN2}. 
Let 
$H_n := H |_{[0,n]}$
be the restriction of 
$H$ 
on 
$L^2 [0,n]$ 
with Dirichlet boundary condition.
Pick  
$0 < \kappa_1 < \kappa_2$ 
arbitrary and set 
\beq
N_{n}(\kappa_1, \kappa_2)
:=
\sharp \{
\mbox{ eigenvalues of 
$H_{n}$ 
in }
(\kappa_1^2, \kappa_2^2)
\}.
\eeq
Since 
the integrated density of states 
$N$ 
of 
$H$ 
is equal to
$N(E) (:= 
\lim_{n \to \infty}
\frac 1n
\sharp \{
\mbox{ eigenvalues of } H_n \le E
\})
=
\pi^{-1} \sqrt{E}$  
as far as 
$\alpha > 0$, 
we have 
\begin{equation}
N_{n} (\kappa_1, \kappa_2)
=
\frac {n}{\pi} 
\left(
\kappa_2 - \kappa_1
\right)
(1 + o(1)), 
\quad
n \to \infty.
\label{N}
\end{equation}
The 
purpose of this paper is to study the 2nd term asymptotics of this equation.
This problem 
is often studied in the context of random matrix theory (e.g., \cite{K}). 
In what follows, 
we state the result which is divided into the following three cases : 
$\alpha > \frac 12$, 
$\alpha = \frac 12$, 
and 
$\alpha < \frac 12$.
\\
\noindent
(1)
{\bf Super-critical decay 
($\alpha > \frac 12$) : }
We need 
to consider suitable subsequences as we did in the study of level statistics : \\
\noindent
{\bf Assumption A }\\
A subsequence 
$\{ n_k \}_{k=1}^{\infty}$
satisfies 
$\lim_{k \to \infty} n_k = \infty$
and as 
$k \to \infty$, 
\beq
\{ \kappa_j n_k \}_{\pi}
=
\gamma_j + o(1), 
\eeq
for some 
$\gamma_j \in [0, \pi)$, 
$j=1,2$.
Here we set 
$\{ x \}_{\pi} 
:=
x- 
\lfloor x \rfloor_{\pi} \cdot \pi$, 
$\lfloor x \rfloor_{\pi}
:=
\left\lfloor x/ \pi \right\rfloor$.
\\
We further 
need to introduce a new quantity. 
Let 
$\theta_t(\kappa)$ 
be the Pr\"ufer angle defined in Section 2.
Set 
$\theta_t(\kappa)= \kappa t + \tilde{\theta}_t(\kappa)$.
Then 
for a.s., the 
$t \to \infty$ 
limit of 
$\tilde{\theta}_t(\kappa)$ 
exists for any 
$\kappa$ 
\cite{KU}.
Let 
$\tilde{\theta}_{\infty}(\kappa)
:=
\lim_{t \to \infty} 
\tilde{\theta}_t (\kappa)$. 
%

%
%
\begin{theorem}
\label{supercritical}
{\bf ($\alpha > \frac 12$)}
Suppose 
Assumption A.
Then 
for a.s., 
\beq
N_{n_k}(\kappa_1, \kappa_2)
-
\left(
\left\lfloor
n_k \kappa_2
\right\rfloor_{\pi}
-
\left\lfloor
n_k \kappa_1
\right\rfloor_{\pi}
\right)
=
\left\lfloor
\gamma_2 + \tilde{\theta}_{\infty}(\kappa_2)
\right\rfloor_{\pi}
-
\left\lfloor
\gamma_1 + \tilde{\theta}_{\infty}(\kappa_1)
\right\rfloor_{\pi}
\eeq
for sufficiently large 
$k$.
\end{theorem}
%
(2)
{\bf Critical decay  ($\alpha = \frac 12$) : }
%
\begin{theorem}
\label{critical}
{\bf ($\alpha=\frac 12$) }\\
Let 
$\{ G (\kappa) \}_{\kappa > 0}$, 
$G$
be mutually independent Gaussian field and a Gaussian such that  
\beq
&&
Cov
\left(
G(\kappa), G(\kappa')
\right)
=
\frac 12
\delta_{\kappa, \kappa'}
\langle [ g_{\kappa}, \overline{g}_{\kappa} ] \rangle, 
\quad
\kappa, \kappa' >0,
\\
&&
Cov
\left( G, G \right)
=
\langle [ g, g ] \rangle, 
\\
&&
g_{\kappa}
=
(L + 2i \kappa)^{-1} F, 
\quad
g :=
L^{-1}( F - \langle F \rangle), 
\\
&&
[ f, g ]
:=
\nabla f \cdot \nabla g.
\eeq
Then as 
$n \to \infty$ 
\beq
&&
\Biggl\{
N_{n}(\kappa_1, \kappa_2)
-
\frac {n}{\pi}
(\kappa_2 - \kappa_1)
-
Re
\left(
\frac {C_1(\kappa_2)}{2\pi \kappa_2}
-
\frac {C_1(\kappa_1)}{2\pi \kappa_1}
\right)
\int_0^{n} a(s)^{2} ds
\Biggr\}
\frac{1}{\sqrt{\log n}}
\\
&&
\stackrel{d}{\to}
\frac {1}{2 \pi \kappa_2} G(\kappa_2)
-
\frac {1}{2  \pi \kappa_1} G(\kappa_1)
-
\left(
\frac {1}{2 \pi  \kappa_2}
-
\frac {1}{2  \pi \kappa_1}
\right)
G
\eeq
in the sense of weak convergence as the processes on 
$(\kappa_1, \kappa_2) 
\in (0, \infty)^2$ 
where 
$C_1(\kappa)
:=
- \frac {i}{2 \kappa}
\langle F g_{\kappa} \rangle$
is a deterministic constant.
\end{theorem}
\begin{remark}
\cite{K} 
studied this problem for CMV matrices, but they do not need to subtract the constant term due to the rotational invariance. 
\end{remark}
%
%
(3)
{\bf Subcritical-decay 
{\bf ($\alpha < \frac 12$)}}
\begin{theorem}
\label{subcritical}
{\bf ($\alpha<\frac 12$) }\\
Set 
$D :=
\min \{
d \in {\bf N} \, | \, 
\frac {1}{2 \alpha} < d+ 1 \}$.
Let 
$\{ G_t (\kappa) \}_{t \in [0,1],\, \kappa > 0}$, 
$\{ G_{t} \}_{t \in [0,1]}$ 
be the mutually independent Gaussian fields such that 
\beq
Cov
\left(
G_t(\kappa), G_s(\kappa')
\right)
&=&
\frac 12
\delta_{\kappa, \kappa'}
\frac {
\langle [ g_{\kappa}, \overline{g}_{\kappa} ] \rangle
}
{1 - 2 \alpha}
( t \wedge s )^{1 - 2 \alpha}
\\
Cov
\left(
G_{t}, G_{s}
\right)
&=&
\frac {
\langle [ g, g ] \rangle
}
{1 - 2 \alpha}
( t \wedge s )^{1 - 2 \alpha}.
\eeq
Then as 
$n \to \infty$  
\beq
&&
\Biggl\{
N_{nt}(\kappa_1, \kappa_2)
-
\frac {nt}{\pi}
(\kappa_2 - \kappa_1)
-
\sum_{j=1}^D
Re
\left(
\frac {C_j(\kappa_2)}{2\pi \kappa_2}
-
\frac {C_j(\kappa_1)}{2\pi \kappa_1}
\right)
\int_0^{nt} a(s)^{j+1} ds
\Biggr\}
\frac{1}{n^{\frac 12 - \alpha}}
\\
&&
\stackrel{d}{\to}
\frac {1}{2 \pi \kappa_2} G_t(\kappa_2)
-
\frac {1}{2 \pi\kappa_1} G_t(\kappa_1)
-
\left(
\frac {1}{2 \pi\kappa_2}
-
\frac {1}{2 \pi \kappa_1}
\right)
G_{t}
\eeq
in the sense of weak convergence 
as the processes for 
$(\kappa_1, \kappa_2, t) \in 
(0, \infty)^2 \times [0,\infty)$, 
where 
$C_j(\kappa)$, 
$j=1, 2, \cdots, D$
are deterministic constants given in 
(\ref{C}). 
\end{theorem}
\begin{remark}
For fixed 
$\kappa_1$, $\kappa_2$, 
RHS is equal to the superposition of Brownian motions in distribution. 
\beq
&&
\frac {1}{2 \pi  \kappa_2} G_t(\kappa_2)
-
\frac {1}{2 \pi  \kappa_1} G_t(\kappa_1)
-
\left(
\frac {1}{2 \pi  \kappa_2}
-
\frac {1}{2 \pi  \kappa_1}
\right)
G_{0 ,t}
\\
&& \qquad
\stackrel{d}{=}
\frac {1}{2 \pi  \kappa_2} 
\sqrt{
\frac 12
\frac {
\langle [ g_{\kappa}, \overline{g}_{\kappa} ] \rangle
}
{1 - 2 \alpha}
}
B_{t^{1 -  2\alpha}}^{(2)}
-
\frac {1}{2 \pi  \kappa_1} 
\sqrt{
\frac 12
\frac {
\langle [ g_{\kappa}, \overline{g}_{\kappa} ] \rangle
}
{1 - 2 \alpha}
}
B_{t^{1 -  2\alpha}}^{(1)}
\\
&& \qquad\qquad
-
\left(
\frac {1}{2 \pi  \kappa_2}
-
\frac {1}{2 \pi \kappa_1}
\right)
\sqrt{
\frac {
\langle [ g, g ] \rangle
}
{1 - 2 \alpha}
}
B_{t^{1 -  2\alpha}}^{(0)}.
\eeq
\end{remark}
\begin{remark}
If 
$a(s)$ 
satisfies 
$a(s) = s^{- \alpha}$, $s \ge R$ 
for some 
$R > 0$, 
we have the following asymptotic expansion. 
\beq
N_{nt}(\kappa_1, \kappa_2)
&\sim&
\frac {nt}{\pi}
(\kappa_2 - \kappa_1)
+
C_2 
(nt)^{1 - 2\alpha}
+
C_3 
(nt)^{1 - 3\alpha}
\\
&& \quad
+
\cdots 
+
C_D 
(nt)^{1 - (D+1)\alpha}
+
n^{\frac 12 - \alpha}
(\mbox{Gaussian}).
\eeq
\end{remark}
\begin{remark}
Theorems \ref{supercritical}, 
\ref{critical}, \ref{subcritical}
roughly imply that the 2nd term in 
(\ref{N}) 
is 
(1) bounded for 
$\alpha > \frac 12$, 
(2) $O(\log n)$ 
for 
$\alpha = \frac 12$, 
and 
(3) $O(n^{1 - 2 \alpha})$ 
for 
$\alpha < \frac 12$. 
That 
the 2nd term grows bigger as 
$\alpha$ 
becomes smaller reflects the fact that the IDS becomes totally different for 
$\alpha = 0$.
\end{remark}
\begin{remark}
We can study 
the case where 
$a(s)$ 
decays slower than 
$s^{-\alpha}$ 
for any 
$\alpha > 0$. 
For instance, if 
$a(s) = (\log s + 1)^{-\delta}$ 
($\delta > 0$), 
then 
\beq
N_n(\kappa_1, \kappa_2)
&=&
\frac {n}{\pi} (\kappa_2 - \kappa_1)
+
N_1(\kappa_1, \kappa_2) + N_2(\kappa_1, \kappa_2), 
\quad
n \to \infty
\eeq
where 
$N_1$
has the following asymptotic expansion
\beq
N_1(\kappa_1, \kappa_2) - 
\sum_{j=2}^k
C_k n (\log n)^{-k \delta}
=
O \left(
n (\log n)^{-(k+1) \delta}
\right)
\eeq
with 
$C_j$ 
being deterministic constants.
$N_2$ 
is a martingale converging to a Gaussian field : 
\[
\frac {N_2(\kappa_1, \kappa_2)}{n^{1/2}(\log n)^{- \delta}} 
\stackrel{d}{\to} 
G(\kappa_1, \kappa_2).
\]
\end{remark}
\begin{remark}
A natural 
and reasonable extension of the problem discussed in this paper is to consider 
\beq
N_n (f) 
&=&
\sum_k
f
\left(
\sqrt{E_k (n)}
\right)
\eeq
where 
$f$ 
is a sufficiently smooth function compactly supported on 
$(0, \infty)$, 
and 
$\{ E_k (n) \}_k$ 
are the positive eigenvalues of 
$H_n$ 
arranged in the increasing order. 
We can show that 
\beq
N_n (f)
&=&
\sum_j 
\int_0^{\infty} f'(\kappa)  g_j(\kappa) d \kappa
\int_0^n a (s)^j ds
+
B_n(f)
+
M_n(f)
\eeq
where 
$g_j(\kappa)$ 
are bounded functions, 
$B_n(f)$ 
is a bounded process, and 
$M_n(f)$
is a martingale.
However, 
we are unable to derive the growth order of 
$\langle M_n (f) \rangle$ 
as 
$n \to \infty$ 
the study of which is postponed to the future work.
\end{remark}
{\bf Decaying coupling constant model (DC model) : }
Let 
us consider the following Hamiltonian : 
\beq
H'_n 
:= - \frac {d^2}{d t^2} + 
\lambda_n F(X_t)
\;
\mbox{ on }
L^2 [0,n], 
\quad
\lambda_n := n^{-\alpha}, 
\quad
\alpha > 0
\eeq
with Dirichlet boundary condition.
Because 
in 1d, the localization length of 
$H = - \triangle + \lambda V$ 
is typically 
$O(\lambda^{- \frac 12})$,  
the property of 
$H'_n$ 
would also change at 
$\alpha = \frac 12$. 
In fact, 
as for the level statistics problem, 
$\xi_L$ 
converges to the (deterministic) clock process for 
$\alpha > \frac 12$, 
Sch$_{\tau}$-process for 
$\alpha = \frac 12$, 
and 
Poisson process for 
$\alpha < \frac 12$
\cite{N2, KN2}.
We can 
apply the discussion in this paper also to  
$H'_n$ 
and obtain the 2nd term asymptotics of 
$N_n(\kappa_1, \kappa_2)$  
under the same notation as in Theorems  
\ref{supercritical}, 
\ref{critical}, 
\ref{subcritical}.
Here 
the major difference from 
the decaying potential model 
$H$  
is that the 2nd term is bounded also for the critical case. 
\begin{theorem}
\label{DCsupercritical}
{\bf ($\alpha > \frac 12$)}
Suppose Assumption A with 
$\gamma_j \ne 0$, 
$j=1,2$. 
Then 
for a.s., 
\beq
N_{n_k}(\kappa_1, \kappa_2)
-
\left(
\left\lfloor
n_k \kappa_2
\right\rfloor_{\pi}
-
\left\lfloor
n_k \kappa_1
\right\rfloor_{\pi}
\right)
=
0
\eeq
for sufficiently large 
$k$.
\end{theorem}
\begin{theorem}
\label{DCcritical}
{\bf ($\alpha = \frac 12$)}
Suppose Assumption A.
We then have 
\beq
N_{n_k}(\kappa_1, \kappa_2)
-
\left(
\left\lfloor
n_k \kappa_2
\right\rfloor_{\pi}
-
\left\lfloor
n_k \kappa_1
\right\rfloor_{\pi}
\right)
\stackrel{d}{\to}
\left\lfloor
\gamma_2 + \tilde{\theta}_{\infty}(\kappa_2)
\right\rfloor_{\pi}
-
\left\lfloor
\gamma_1 + \tilde{\theta}_{\infty}(\kappa_1)
\right\rfloor_{\pi}
\eeq
as 
$k \to \infty$. 
\end{theorem}
\begin{theorem}
\label{DCsubcritical}
{\bf ($\alpha<\frac 12$) }
Set 
$D :=
\min \{
d \in {\bf N} \, | \, 
\frac {1}{2 \alpha} < d+ 1 \}$.
Let 
$\{ G_t (\kappa) \}_{t \in [0,1], \kappa > 0}$, 
$\{ G_{t} \}_{t \in [0,1]}$
be the mutually independent Gaussians such that 
\beq
Cov
\left(
G_t(\kappa), G_s(\kappa')
\right)
&=&
\frac 12
\delta_{\kappa, \kappa'}
\langle [ g_{\kappa}, \overline{g}_{\kappa} ] \rangle
( t \wedge s )^{1 - 2 \alpha}
\\
Cov
\left(
G_{t}, G_{s}
\right)
&=&
\langle [ g, g ] \rangle
( t \wedge s )^{1 - 2 \alpha}. 
\eeq
We then have 
\beq
&&
\Biggl\{
N_{nt}(\kappa_1, \kappa_2)
-
\frac {nt}{\pi}
(\kappa_2 - \kappa_1)
-
\sum_{j=1}^D
Re
\left(
\frac {C_j(\kappa_2)}{2\pi \kappa_2}
-
\frac {C_j(\kappa_1)}{2\pi \kappa_1}
\right)
(nt)^{1-(j+1)\alpha}
%
\Biggr\}
\frac{1}{n^{\frac 12 - \alpha}}
\\
&&
\stackrel{d}{\to}
\frac {1}{2 \pi \kappa_2} G_t(\kappa_2)
-
\frac {1}{2 \pi\kappa_1} G_t(\kappa_1)
-
\left(
\frac {1}{2 \pi\kappa_2}
-
\frac {1}{2 \pi \kappa_1}
\right)
G_{t}.
\eeq
\end{theorem}
The 
main ingredient of the proof is to express 
$N_{nt}(\kappa_1, \kappa_2)$ 
in terms of the Pr\"ufer angles
$\theta_t(\kappa)$, 
as is done in 
\cite{K}. 
Then 
we study the behavior of 
$\theta_t(\kappa)$ 
by the martingale analysis developed in \cite{KU}. 
The plan 
of this paper is as follows.
In Section 2, 
we introduce the Pr\"ufer variable and compute the basic integrals which frequently appears in this paper. 
In Sections 3,4,5, 
we prove Theorems 
\ref{supercritical},
\ref{critical}
and 
\ref{subcritical} 
respectively. 
%

\section{Preliminaries}
\subsection{Pr\"ufer coordinate}
For 
$\kappa > 0$ 
let 
$x_t (\kappa)$
be the solution to the Schr\"odinger equation 
$H x_t = \kappa^2 x_t$, 
$x_0 (\kappa) = 0$ 
which we represent in terms of the Pr\"ufer coordinate : 
\beq
\left(
\begin{array}{c}
x_t(\kappa) \\ x'_t(\kappa) / \kappa
\end{array}
\right)
=
\left(
\begin{array}{c}
r_t \sin \theta_t(\kappa)
\\
r_t \cos \theta_t(\kappa)
\end{array}
\right), 
\quad
\theta_0 (\kappa) = 0.
\eeq
Let 
$\theta_t(\kappa) 
=
\kappa t + \tilde{\theta}_t(\kappa)$.
By 
Sturm's oscillation theorem, 
\begin{equation}
N_{nt}(\kappa_1, \kappa_2) 
- 
\frac {1}{\pi}nt (\kappa_2 - \kappa_1)
=
\frac {1}{\pi}
\left(
\tilde{\theta}_{nt}(\kappa_2) 
- 
\tilde{\theta}_{nt}(\kappa_1)
\right)
\pm 1
\label{plusminus}
\end{equation}
so that it suffices to study the behavior of 
$\tilde{\theta}_t(\kappa)$.
Noting that 
$\tilde{\theta}_t(\kappa)$ 
satisfies the following integral equation, 
\beq
\tilde{\theta}_{t}(\kappa)
&=&
\frac {1}{2 \kappa}
Re
\int_0^t
\left(
e^{2 i \theta_s(\kappa)}- 1
\right)
a(s) F(X_s) ds, 
\eeq
we set 
\beq
J_t^{(n)}(\kappa)
&:=&
\int_0^{nt}
a(s)e^{2i \theta_{s}(\kappa)} 
F(X_s) ds
\\
J_{0,t}^{(n)}
&:=&
\int_0^{nt} a(s) F(X_s)ds.
\eeq
We 
can then decompose 
\begin{eqnarray}
&&
\tilde{\theta}_{nt}(\kappa_2)
-
\tilde{\theta}_{nt}(\kappa_1)
=
A + B + C
\label{decomposition}
\\
&& \qquad
A 
=
\frac {1}{2 \kappa_2} Re \;
J_t^{(n)}(\kappa_2), 
\quad
B =
-\frac {1}{2 \kappa_1} Re \;
J_t^{(n)}(\kappa_1),
\nonumber
\\
&& \qquad
C =
-\left(
\frac {1}{2\kappa_2} - \frac {1}{2\kappa_1}
\right)
Re \; J_{0,t}^{(n)}.
\nonumber
\end{eqnarray}
We shall 
study the behavior of 
$J_t^{(n)}$, 
$J_{0, t}^{(n)}$, 
as 
$n \to \infty$. 
%
\subsection{Basic Calculus on Integrals}
For 
$H \in C^{\infty}(M)$, 
set 
\[
K^{(n)}_{m, \beta, \kappa, t} (H)
:=
\int_0^{nt}
a(s)^m 
e^{i \beta \theta_{s}(\kappa)}
H(X_s) ds, 
\quad
m \in {\bf N}, 
\;
\beta \in {\bf R}.
\]
\begin{lemma}
\label{CorKm}
If 
$\beta \ne 0$,
\begin{eqnarray}
K^{(n)}_{m, \beta, \kappa, t} (H)
&=&
K^{(n)}_{m+1, \beta+2, \kappa, t}
\left(
T_{\beta, \kappa}^+(H)
\right)
+
K^{(n)}_{m+1, \beta-2, \kappa, t}
\left(
T_{\beta, \kappa}^- (H)
\right)
\nonumber
\\
%
&&\quad
+
K^{(n)}_{m+1, \beta, \kappa, t}
\left(
T_{\beta, \kappa}^0 (H)
\right)
+
L_{m, \beta, \kappa, t}^{(n)}
+
M_{m, \beta, \kappa, t}^{(n)}.
\label{decompositionK}
\end{eqnarray}
where 
$T_{\beta}^{\sharp}$, 
$\sharp = \pm, 0$
are the operators acting on 
$C^{\infty}(M)$ 
defined by 
\beq
(T_{\beta, \kappa}^+ H)(x)
&=&
-\frac {i \beta}{2 \kappa} 
\cdot
\frac 12
\cdot
F(x) 
\cdot
\left(
R_{\beta \kappa}H
\right) (x)
\\
(T_{\beta, \kappa}^- H)(x)
&=&
-\frac {i \beta}{2 \kappa} 
\cdot
\frac 12
\cdot
F(x) 
\cdot
\left(
R_{\beta \kappa}H
\right) 
(x)
\\
(T_{\beta, \kappa}^0 H)(x)
&=&
\frac {i \beta}{2 \kappa} 
\cdot
F(x) 
\cdot
\left(
R_{\beta \kappa}H
\right)
(x)
\\
R_{\kappa}H 
&:=&
(L + i \kappa)^{-1}H.
\eeq
$L_{m, \beta, \kappa, t}^{(n)}$
is bounded and 
$M_{m, \beta, \kappa, t}^{(n)}$
is a martingale such that 
\beq
&&
L_{m, \beta, \kappa, t}^{(n)}
=
\left[
a(s)^m e^{i \beta \theta_{s}(\kappa)}
R_{\beta \kappa} (H)(X_s) ds
\right]_0^{nt}
\\
&& \qquad\qquad\qquad
- 
\int_0^{nt}
(a(s)^m)' 
e^{i \beta \theta_{s}(\kappa)}
R_{\beta \kappa}(H)(X_s) ds
\\
&&
M_{m, \beta, \kappa, t}^{(n)}
=
- 
\int_0^{nt}
a(s)^m e^{i \beta \theta_s(\kappa)}
(\nabla R_{\beta \kappa}H)(X_s) d X_s
\\
&&
\langle M_{m, \beta, \kappa, t}^{(n)}, 
M_{m, \beta, \kappa, t}^{(n)}
\rangle, 
\;
\langle M_{m, \beta, \kappa, t}^{(n)}, 
\overline{
M_{m, \beta, \kappa, t}^{(n)}
}
\rangle
=
O\left(
\int_0^{nt} a(s)^{2m} ds
\right), 
\quad
n \to \infty.
\eeq
\end{lemma}
\begin{remark}
It is 
not necessary to distinguish  
$T_{\beta, \kappa}^+$
from  
$T_{\beta, \kappa}^-$.
We put 
the plus minus symbol 
$\pm$ 
only to facilitate the computation of some combinatorial quantities in the proof of Proposition \ref{D}.

\end{remark}
\begin{proof}
By Ito's formula, 
\beq
e^{i \kappa s} H(X_s) ds
=
d \left(
e^{i \kappa s}
R_{\kappa} (H)
\right)
-
e^{i \kappa s} 
\nabla R_{\kappa} (H) d X_s
\eeq
which we substitute into 
$K^{(n)}_{m, \beta, \kappa, t}$ 
and integrate by parts. 
\beq
K_{m, \beta, \kappa, t}^{(n)}(\kappa)
&=&
\left[
a(s)^m e^{i \beta \theta_{s}(\kappa)}
R_{\beta \kappa} (H)(X_s) 
\right]_0^{nt}
\\
&& - 
\int_0^{nt}
(a(s)^m)' 
e^{i \beta \theta_{s}(\kappa)}
R_{\beta \kappa}(H)(X_s) ds
\\
&& - 
\frac {i \beta}{2 \kappa}
\int_0^{nt}
Re \left(
e^{2i \theta_{s}(\kappa)}-1 
\right)
e^{i \beta \theta_{s}(\kappa)}
a(s)^{m+1}
F(X_s) R_{\beta \kappa}(H)(X_s) ds
\\
&& -
\int_0^{nt} a(s)^m 
e^{i\beta \theta_{s}(\kappa)}
\nabla R_{\beta \kappa} (H)(X_s) d X_s
\\
&=:& 
K_1(\kappa) + \cdots + K_4 (\kappa).
\eeq
$K_1$, $K_2$
are bounded. 
$K_3$ 
gives the first three terms in the RHS of 
(\ref{decompositionK}).
$K_4$ 
is a martingale and it is easy to check the statement for those. 
\QED
\end{proof}
%
%
%
\begin{lemma}
\label{CorK0}
If 
$\beta = 0$, 
\beq
K_{m, 0, t}^{(n)}(H)
&=&
\langle H \rangle 
\int_0^{nt} a(s)^m ds
+
L_{m, 0, t}^{(n)}
+
M_{m, 0, t}^{(n)}
\eeq
where 
$L_{m, 0, t}^{(n)}$ 
is bounded, and 
$M_{m, 0, t}^{(n)}$
is a martingale such that 
\beq
L_{m, 0, t}^{(n)}
&=&
\left[
a(s)^m (RH)(X_s) 
\right]_0^{nt}
-
\int_0^{nt} ( a(s)^m )' (RH)(X_s) ds
\\
\langle M_{m, 0, t}^{(n)} \rangle
&=&
O \left(
\int_0^{nt} a(s)^{2m} ds
\right), 
\quad
n \to \infty
\eeq
where 
$(RH)(s) := L^{-1}(H - \langle H \rangle)(s)$.
\end{lemma}
\begin{proof}
By 
Ito's formula, we have 
\beq
H(X_s) ds
&=&
\langle H \rangle ds
+
d \left(
R(H)(X_s)
\right)
-
\nabla (R(H))(X_s) d X_s
\eeq
which gives 
\beq
K_{m, 0}^{(n)}
&=&
\langle H \rangle
\int_0^{nt} a(s)^m ds
+
\left[
a(s)^m G(X_s) 
\right]_0^{nt}
\\
&&-
\int_0^{nt} ( a(s)^m )' G(X_s) ds
- 
\int_0^{nt} a(s)^m \nabla G(X_s) d X_s
\\
&=:& K_1' + \cdots + K'_4. 
\eeq
$K'_2$, $K'_3$ 
are bounded and 
$K'_4$
is a martingale. 
\QED
\end{proof}
%
\subsection{Expansion of $J$}
In this subsection 
we study the behavior of 
$J_t^{(n)}$, 
$J_{0,t}^{(n)}$
by using Lemmas \ref{CorKm}, \ref{CorK0}.
%
%
\begin{proposition}
\label{D}
For any 
$D \ge 1$ 
we have 
\begin{equation}
J_t^{(n)}(\kappa)
=
\sum_{k=1}^D
C_k(\kappa)  
\int_0^{nt} a(s)^{k+1} ds
+
K_D(\kappa)
+
L_D(\kappa) 
+ 
M_D(\kappa)
\label{expansionJ}
\end{equation}
where 
\[
K_D(\kappa)
=
O \left(
\int_0^{nt} a(s)^{D+2} ds
\right)
\]
and 
$L_D$
is bounded.
The constants 
$C_k(\kappa)$
are given in (\ref{C}) below. 
$M_D (\kappa)$
is a martingale such that 
\beq
M_D (\kappa)
&=&
M_{1, 2, \kappa, t}
+
M'_D
\\
M_{1, 2, \kappa, t}
&:=&
-
\int_0^{nt}
a(s) e^{2i \theta_s(\kappa)}
\nabla g_{\kappa}(X_s) d X_s
\\
\langle M_D(\kappa), M_D(\kappa) \rangle
&=&
O \left(
\int_0^{nt} a(s)^4 ds
\right)
\\
\langle M_D (\kappa), \overline{M}_D(\kappa) \rangle
&=&
\langle 
M_{1,2,\kappa, t}, \overline{M}_{1,2,\kappa, t} 
\rangle
(1+o(1))
\\
&=&
\langle
[ g_{\kappa}, \overline{g}_{\kappa} ] 
\rangle
\int_0^{nt} 
a(s)^2 ds
(1 + o(1))
\eeq
where we set 
$g_{\kappa} := R_{2 \kappa} F$.
\end{proposition}
\begin{proof}
\mbox{}\\
(1)
1st step : 
Letting 
$m=1$, $\beta = 2$ 
in Lemma \ref{CorKm}, 
\begin{eqnarray}
J_{t}^{(n)}(\kappa)
&=&
K_{1, 2, \kappa, t}^{(n)}(F)
\nonumber
\\
&=&
K_{2, 4, \kappa, t}^{(n)}
\left(
T_{2, \kappa}^+ (F)
\right) 
+ 
K_{2, 0, \kappa, t}^{(n)}
\left(
T_{2, \kappa}^- (F)
\right) + 
K_{2, 2, \kappa, t}
\left(
T_{2, \kappa}^0(F)
\right)
\nonumber
\\
&& \quad
+ L_{1, 2, \kappa, t}^{(n)} 
+ M_{1, 2, \kappa, t}^{(n)}
\label{firststep}
\end{eqnarray}
We further use 
Lemma \ref{CorKm} to the 1st and 3rd terms in the RHS of 
(\ref{firststep}), 
so that they are 
$O\left(
\int_0^{nt} a(s)^3 ds
\right)$.
For the 2nd term 
$K_{2, 0, \kappa, t}^{(n)}$, 
we use Lemma 
\ref{CorK0}.
\beq
&&
K_{2, 0, \kappa, t}^{(n)}
\left(
T_{2, \kappa}^- (F)
\right)
=
\langle 
T_{2, \kappa}^- (F)
\rangle
\int_0^{nt} a(s)^2 ds
+
L_{2, 0, t}^{(n)}
+
M_{2, 0, t}^{(n)}, 
\\
&&
\langle 
M_{2, 0, t}^{(n)}, 
M_{2, 0, t}^{(n)}
\rangle, 
\;
\langle 
M_{2, 0, t}^{(n)}, 
\overline{ M_{2, 0, t}^{(n)} }
\rangle
=
O
\left(
\int_0^{nt} a(s)^4 ds 
\right).
\eeq
For the 5th martingale term 
$M_{1, 2, \kappa, t}^{(n)}$ 
in the RHS of (\ref{firststep}), 
we estimate its quadratic variation by Lemmas 
\ref{CorKm}, \ref{CorK0}. 
\beq
M_{1, 2, \kappa, t}^{(n)}
&=&
-
\int_0^{nt}
a(s) e^{2i \theta_s(\kappa)} 
(\nabla g_{\kappa})(X_s) dX_s
\\
\langle M_{1, 2, \kappa, t}^{(n)}, 
M_{1, 2, \kappa, t}^{(n)}
\rangle
&=&
K_{2, 4, \kappa, t} (\varphi_{\kappa}), 
\quad
\varphi_{\kappa}
:=
[ g_{\kappa}, g_{\kappa} ] 
\\
&=&
K_{3, 6, \kappa, t}(T^+_{\beta, \kappa}(\varphi_{\kappa}))
+
K_{3, 2, \kappa, t}(T^-_{\beta, \kappa}(\varphi_{\kappa}))
+
K_{3, 4, \kappa, t}(T_{\beta, \kappa}(\varphi_{\kappa}))
\\
&& + 
L_{2, 4, \kappa, t}
+
M_{2, 4, \kappa, t}.
\eeq
For the 
first three terms of RHS we use Lemma 
\ref{CorKm} 
again. 
Moreover 
we have 
$\langle M_{2, 4, \kappa, t}, M_{2, 4, \kappa, t} \rangle
=
O \left(
\int_0^{nt} a(s)^4 ds 
\right)
=
O(n^{1- 4\alpha})$.
Since 
$1 - 3 \alpha > \frac {1 - 4 \alpha}{2}$
if and only if 
$\alpha < \frac 12$, 
it is lower order or bounded. 
Therefore
\beq
\langle M_{1, 2, \kappa, t}^{(n)}, 
M_{1, 2, \kappa, t}^{(n)}
\rangle
&=&
O \left(
\int_0^{nt} a(s)^4 ds
\right)
\\
\langle M_{1, 2, \kappa, t}^{(n)}, 
\overline{M}_{1, 2, \kappa, t}^{(n)}
\rangle
&=&
\int_0^{nt}
a(s)^2
[ g_{\kappa}, \overline{g}_{\kappa} ] (X_s) ds
\\
&=&
\langle 
[ g_{\kappa}, \overline{g}_{\kappa} ]
\rangle
\int_0^{nt} a(s)^2 ds
+
\mbox{ (bounded) }
+
M'_{1,2,\kappa, t}
\\
\langle M'_{1,2,\kappa, t}, M'_{1,2,\kappa, t} \rangle
&=&
O \left(
\int_0^{nt} a(s)^4 ds
\right)
\eeq
and 
(\ref{expansionJ})
is proved for 
$D=1$.\\
\noindent
(2)
$(k+1)$-th step : 
we iterate this process.
After the $k$-th step, we have a sum of 
$\sum_{j=1}^k 
C_j(\kappa)
\int_0^{nt} a(s)^{j+1} ds$, 
$K_{k+1, \beta, \kappa, t}(H)$
$(\beta \ne 0)$, 
bounded term, and a sum of martingales. 
So in the $(k+1)$-th step, 
we apply Lemma \ref{CorKm} to 
$K_{k+1, \beta, \kappa, t}(H)
\;(\beta \ne 0)$ 
to have 
$K_{k+2, \;\beta', \;\kappa,\; t}(H)
(\beta' \ne 0)$,
$K_{k+2, \,0, \;\kappa,\; t}(H)$, 
bounded term and a sum of martingales.
We further apply 
Lemma \ref{CorK0} to 
$K_{k+2, 0, \kappa, t}(H)$ 
to have a deterministic term proportional to 
$\int_0^{nt} a(s)^{k+2} ds$ 
and a sum of bounded terms and martingales.
We note that 
the martingales which emerge in each steps have quadratic variation with order at most 
$O\left(
\int_0^{nt} a(s)^4 ds
\right)$, 
except 
$M_{1,2, \kappa, t}$.
Letting 
$M_K$ 
be the sum of all martingales appeared up to the 
$D$-th step, 
$M_K$ 
satisfies the statement in Proposition \ref{D}. 
To compute 
the coefficients proportional to 
$\int_0^{nt} a(s)^{k+1} ds$, 
we consider a set of indices :
\beq
S_k :=
\Biggl\{
\left(
(\epsilon_1, \cdots, \epsilon_{k-1}), 
(\beta_1, \cdots, \beta_{k-1})
\right)
\, \Biggl| \, 
&&
\epsilon_i = 0, \pm 1, 
\\
&&
\beta_i \in 2 {\bf N}, 
\beta_{i+1} = \beta_i + 2 \epsilon_i, 
\beta_1 = 2
\\
&&
\sum_{i=1}^j \epsilon_i \ge 0, 
1 \le j \le k-2, 
\sum_{i=1}^{k-1} \epsilon_i = 0
\Biggr\}
\eeq
then the desired coefficients is given by 
\begin{equation}
C_k (\kappa) = 
\left\{
\begin{array}{cc}
\langle T_{2, \kappa}^{-1} F \rangle & (k=1) \\
\sum_{
\left(
(\epsilon_i), (\beta_i)
\right) 
\in S_k
}
\langle
T_{2, \kappa}^{-1} 
T_{\beta_{k-1}, \kappa}^{\epsilon_{k-1}}
\cdots
T_{\beta_1, \kappa}^{\epsilon_1} F
\rangle
&
(k \ge 2)
\end{array}
\right.
\label{C}
\end{equation}
For instance, 
omitting the 
$\kappa$
-dependence, we have 
\beq
C_2 
&=&
\langle 
T_2^{-1} T_2^0 F
\rangle
\\
C_3 
&=&
\langle 
T_2^- T_4^- T_2^+ F
+
T_2^- T_2^0 T_2^0 F
\rangle.
\eeq
Proof of Proposition \ref{D} is now complete. 
\QED
\end{proof}
Lemma \ref{CorK0} 
yields the following decomposition of 
$J^{(n)}_{0, t}$. 
\begin{proposition}
\label{D0}
\begin{eqnarray}
J_{0, t}^{(n)}
&=&
L_{0, nt} + M_{0, nt}
\label{expansionJ0}
\\
L_{0, nt}
&=&
- a(0) g(X_0) - 
\int_0^{nt} a'(s) g(X_s) ds
\nonumber
\\
M_{0, nt}
&=&
-\int_0^{nt}
a(s) \nabla g(X_s) d X_s
\nonumber
\end{eqnarray}
where 
$L_{0, nt}$
is bounded, 
$M_{0, nt}$
is a martingale such that  
\beq
\langle M_0, M_0 \rangle
&=&
\langle [g, g] \rangle
\int_0^{nt} a(s)^2 ds (1+o(1)).
\eeq
\end{proposition}
%
%
\section{Proof 
for supercritical case
}
First of all, 
we notice that the argument 
of the proof of Proposition 7.1 in \cite{KN1} shows that the distribution of 
$\tilde{\theta}_{\infty}(\kappa)$ 
is continuous for 
$\alpha > \frac 12$
(also for DC model with 
$\alpha \ge \frac 12$). 
In fact, 
we can show that 
$\lim_{m \to \infty}
\lim_{t \to \infty}
{\bf E}[ e^{im \tilde{\theta}_t (\kappa)} ] = 0$.
Thus 
$\left\{
\gamma_j + \tilde{\theta}_{\infty}(\kappa_j) 
\right\}_{\pi} \ne 0$, 
a.s. so that 
\beq
\theta_{n_k}(\kappa_j)
&=&
\lfloor
\kappa_{j} n_k 
\rfloor_{\pi}
\pi
+
\lfloor 
\gamma_j 
+
\tilde{\theta}_{\infty}(\kappa_j)
\rfloor_{\pi}
\pi
+
\left\{
\gamma_j + \tilde{\theta}_{\infty}(\kappa_j) 
\right\}_{\pi}
+
o(1), 
\quad
j=1,2,
\quad
a.s.
\eeq
Theorem \ref{supercritical} 
now follows from Sturm's oscillation theorem.
\QED
%

\section{
Proof for critical case
}
Using 
Proposition \ref{D} 
with 
$D=1$ 
and substituting it into 
(\ref{decomposition})
yields 
\beq
&&
\tilde{\theta}_{n}(\kappa_2)
-
\tilde{\theta}_{n}(\kappa_1)
\\
&=&
\left(
\frac {Re C_1 (\kappa_2)}{2 \kappa_2}
-
\frac {Re C_1 (\kappa_1)}{2 \kappa_1}
\right)
\int_0^{n} a(s)^{2} ds
\\
&& + 
\frac {1}{2 \kappa_2} Re (K_1(\kappa_2)+L_1(\kappa_2))
-
\frac {1}{2 \kappa_1} Re (K_1(\kappa_1)+L_1(\kappa_1)))
-
\left(
\frac {1}{2 \kappa_2}
-
\frac {1}{2 \kappa_1}
\right)
L_0
\\
&& +
\frac {1}{2 \kappa_2} Re M_1(\kappa_2)
-
\frac {1}{2 \kappa_1} Re M_1 (\kappa_1)
-
\left(
\frac {1}{2 \kappa_2}
-
\frac {1}{2 \kappa_1}
\right)
M_0.
\eeq
By 
Lemmas \ref{CorKm}, \ref{CorK0}, 
\beq
K_1 (\kappa)
&=&
O \left(
\int_0^n a(s)^4 ds
\right)
\\
\langle
Re M_1(\kappa), Re M_1 (\kappa')
\rangle
&=&
\frac 12
\langle 
[ g_{\kappa}, 
\overline{g}_{\kappa} ]
\rangle
\log n
(\delta_{\kappa, \kappa'}
+ o(1))
\\
\langle 
Re M_1, M_0 
\rangle
&=&
o(\log n)
\\
\langle 
M_0, M_0 
\rangle
&=&
\langle 
[ g,g ]
\rangle
\log n
(1 + o(1)).
\eeq
Let 
\beq
M(\kappa_1, \kappa_2) 
&:=&
\frac {1}{2 \kappa_2} Re M_1(\kappa_2)
-
\frac {1}{2 \kappa_1} Re M_1 (\kappa_1)
-
\left(
\frac {1}{2 \kappa_2}
-
\frac {1}{2 \kappa_1}
\right)
M_0
\eeq
be its martingale part.
Let 
$G (\kappa)$, 
$G$
be independent Gaussians satisfying the covariance condition stated in Theorem \ref{critical}. 
Then 
by the martingale central limit theorem, 
\beq
\frac {M(\kappa_1, \kappa_2)}
{\sqrt{\log n}}
\stackrel{d}{\to}
\frac {1}{2 \kappa_2} G(\kappa_2)
-
\frac {1}{2 \kappa_1} G(\kappa_1)
-
\left(
\frac {1}{2 \kappa_2}
-
\frac {1}{2 \kappa_1}
\right)
G
\eeq
which leads us to the completion of proof : 
\beq
&&
\Biggl\{
\tilde{\theta}_{n}(\kappa_2)
-
\tilde{\theta}_{n}(\kappa_1)
-
Re
\left(
\frac {C_1(\kappa_2)}{2 \kappa_2}
-
\frac {C_1(\kappa_1)}{2 \kappa_1}
\right)
\int_0^{n} a(s)^{2} ds
\Biggr\}
\frac{1}{\sqrt{\log n}}
\\
&& \qquad 
=
\Biggl\{
\frac {1}{2 \kappa_2} 
Re \left(
K_1(\kappa_2) + L_1(\kappa_2)
\right)
-
\frac {1}{2 \kappa_1} 
Re \left(
K_1(\kappa_1) + L_1 (\kappa_1)
\right)
\\
&& \qquad\qquad
-
\left(
\frac {1}{2 \kappa_2}
-
\frac {1}{2 \kappa_1}
\right)
L_0
\Biggr\}
\frac{1}{\sqrt{\log n}}
+ 
\frac {M(\kappa_1, \kappa_2)}
{\sqrt{\log n}}
\\
&&\qquad
\stackrel{d}{\to}
\frac {1}{2 \kappa_2} G(\kappa_2)
-
\frac {1}{2 \kappa_1} G(\kappa_1)
-
\left(
\frac {1}{2 \kappa_2}
-
\frac {1}{2 \kappa_1}
\right)
G.
\eeq
\QED
%

\section{
Proof for subcritical case
}
Let 
$D :=
\min \{
d \in {\bf N} \, | \, 
\frac {1}{2 \alpha} < d+ 1 \}$.
Substituting 
(\ref{expansionJ}), (\ref{expansionJ0}) 
into
(\ref{decomposition}), 
yields
\beq
&&
\tilde{\theta}_{nt}(\kappa_2)
-
\tilde{\theta}_{nt}(\kappa_1)
\\
&=&
\sum_{j=1}^D
\left(
\frac {Re C_j(\kappa_2)}{2 \kappa_2}
-
\frac {Re C_j(\kappa_1)}{2 \kappa_1}
\right)
\int_0^{nt} a(s)^{j+1} ds
\\
&& + 
\frac {1}{2 \kappa_2} Re (K_D(\kappa_2)+L_D(\kappa_2))
-
\frac {1}{2 \kappa_1} Re (K_D(\kappa_1)+L_D(\kappa_1))
-
\left(
\frac {1}{2 \kappa_2}
-
\frac {1}{2 \kappa_1}
\right)
L_0
\\
&& +
\frac {1}{2 \kappa_2} Re M_D(\kappa_2)
-
\frac {1}{2 \kappa_1} Re M_D (\kappa_1)
-
\left(
\frac {1}{2 \kappa_2}
-
\frac {1}{2 \kappa_1}
\right)
M_0.
\eeq
Let 
\beq
M(\kappa_1, \kappa_2) 
&:=&
\frac {1}{2 \kappa_2} Re M_D(\kappa_2)
-
\frac {1}{2 \kappa_1} Re M_D (\kappa_1)
-
\left(
\frac {1}{2 \kappa_2}
-
\frac {1}{2 \kappa_1}
\right)
M_0
\eeq
be the martingale part. 
We estimate 
the quadratic variations of 
$M_D(\kappa_2)$, 
$M_D(\kappa_1)$, 
$M_0$ 
by using 
Lemmas \ref{CorKm}, \ref{CorK0}. 
\beq
\langle
Re M_D(\kappa), Re M_D (\kappa')
\rangle
&=&
\frac 12
\langle 
[ g_{\kappa}, 
\overline{g}_{\kappa} ]
\rangle
\frac {n^{1 - 2 \alpha}}{1 - 2 \alpha}
t^{1 - 2 \alpha}
(\delta_{\kappa, \kappa'} + o(1))
\\
\langle 
Re M_D, M_0 
\rangle
&=&
o(n^{1 - 2 \alpha})
\\
\langle 
M_0, M_0 
\rangle
&=&
\langle 
[ g,g ]
\rangle
\frac {n^{1 - 2 \alpha}}{1 - 2 \alpha}
t^{1 - 2 \alpha}
(1 + o(1)).
\eeq
Thus  
letting 
$\{ G_t(\kappa) \}$, $\{ G_t \}$ 
be the Gaussians defined in the statement of Theorem \ref{subcritical}, 
we have 
\beq
\frac {M(\kappa_1, \kappa_2)}
{n^{\frac 12 - \alpha}}
\stackrel{d}{\to}
\frac {1}{2 \kappa_2} G_t(\kappa_2)
-
\frac {1}{2 \kappa_1} G_t(\kappa_1)
-
\left(
\frac {1}{2 \kappa_2}
-
\frac {1}{2 \kappa_1}
\right)
G_{t}.
\eeq
On the other hand, since 
\beq
K_D
&=&
O \left(
\int_0^{nt} a(s)^{D+2} ds
\right)
=
O \left(
n^{1 - (D+2) \alpha}
\right)
\eeq
$K_D$ 
is lower order compared to the martingale terms.
Therefore
\beq
&&
\Biggl\{
\tilde{\theta}_{nt}(\kappa_2)
-
\tilde{\theta}_{nt}(\kappa_1)
-
\sum_{j=1}^D
Re
\left(
\frac {C_j(\kappa_2)}{2 \kappa_2}
-
\frac {C_j(\kappa_1)}{2 \kappa_1}
\right)
\int_0^{nt} a(s)^{j+1} ds
\Biggr\}
\frac{1}{n^{\frac 12 - \alpha}}
\\
&&\qquad 
=
\Biggl\{
\frac {1}{2 \kappa_2} 
Re 
\left(
K_D(\kappa_2) + L_D(\kappa_2)
\right)
-
\frac {1}{2 \kappa_1} Re 
\left(
K_D(\kappa_1) + L_D(\kappa_1)
\right)
\\
&& \qquad\qquad
-
\left(
\frac {1}{2 \kappa_2} - 
\frac {1}{2 \kappa_1}
\right)
L_0
\Biggr\}
\frac {1}{n^{\frac 12 - \alpha}}
+ 
\frac {M(\kappa_1, \kappa_2)}
{n^{\frac 12 - \alpha}}
\\
&&\qquad
\stackrel{d}{\to}
\frac {1}{2 \kappa_2} G_t(\kappa_2)
-
\frac {1}{2 \kappa_1} G_t(\kappa_1)
-
\left(
\frac {1}{2 \kappa_2}
-
\frac {1}{2 \kappa_1}
\right)
G_{0 ,t}
\eeq
completing the proof of Theorem \ref{subcritical}.
\QED
%


\vspace*{1em}
\noindent {\bf Acknowledgement }
The author 
would like to thank the Isaac Newton Institute for
Mathematical Sciences for its hospitality during the programme
``Periodic and Ergodic Spectral Problems"
supported by EPSRC Grant Number EP/K032208/1,  
and also to the referee for many useful comments. 
This work is partially supported by 
JSPS KAKENHI Grant Number 26400145.

%

\end{document}